\documentclass[letter,11pt]{article}
\usepackage[margin=1in]{geometry}
\usepackage[numbers,sort]{natbib}
\usepackage{authblk}

\usepackage{amsmath}\allowdisplaybreaks
\usepackage{amssymb}
\usepackage{amsthm}
\usepackage{dsfont}
\usepackage{mathrsfs}
\usepackage[all]{xy}
\usepackage{empheq}
\usepackage{algorithm,algorithmic}
\usepackage{multirow}
\usepackage{multicol}
\usepackage{verbatim}
\usepackage{fancybox,framed}
\usepackage{array}
\usepackage{xcolor}
\usepackage{thm-restate}

\usepackage{hyperref}
\hypersetup{
	colorlinks   = true, 
	urlcolor     = blue, 
	linkcolor    = blue, 
	citecolor   = blue 
}
\theoremstyle{theorem}
\newtheorem{thm}{Theorem}[section]
\newtheorem{thrm}{Theorem}

\newtheorem{prop}[thrm]{Proposition}
\newtheorem{Def}[thrm]{Definition}

\newcommand{\E}{\mathop{\mathbb{E}}}

\newcommand{\R}{\mathbb{R}}

\newcommand{\cP}{\mathcal{P}}

\newcommand{\cW}{\mathcal{W}}
\newcommand{\argmax}{\mathrm{argmax}}

\newcommand{\dom}{{\rm dom}}

\newcommand{\G}{{\rm Greedy}}
\newcommand{\fps}{f^{+}}

\newcommand{\fa}{\sigma}

\newcommand{\OPT}{\mathrm{OPT}}
\newcommand{\Var}{\mathrm{Var}}
\newcommand{\Evt}{\mathcal{E}}

\newcommand{\aL}{\mathcal{L}}

\newcommand{\bx}{\boldsymbol{x}}
\newcommand{\bI}{\boldsymbol{I}}
\newcommand{\mF}{\mathcal{F}}

\title{On Adaptivity Gaps of Influence Maximization 
	under the Independent Cascade Model with Full Adoption Feedback}

\author[1]{Wei Chen\thanks{Email address: \url{weic@microsoft.com}.}\thanks{Supported in part by the National Natural Science Foundation of China (Grant No. 61433014).}}
\author[2]{Binghui Peng\thanks{Email address: \url{pbh15@mails.tsinghua.edu.cn}.}}

\affil[1]{Microsoft Research}
\affil[2]{Tsinghua University}

\date{\vspace{-1cm}}

\begin{document}
	
	\maketitle
	
	\begin{abstract}
	In this paper, we study the {\em adaptivity gap} of the influence maximization problem under independent cascade model when {\em full-adoption} feedback is available. 
	Our main results are to derive upper bounds on several families of well-studied influence graphs, including in-arborescences, out-arborescences and bipartite graphs.
	Especially, we prove that the adaptivity gap for the in-arborescence is between $[\frac{e}{e-1}, \frac{2e}{e - 1}]$ and for the out-arborescence, the gap is between $[\frac{e}{e-1}, 2]$. 
	These are the first constant upper bounds in the full-adoption feedback model.
	We provide several novel ideas to tackle with correlated feedback appearing in the adaptive stochastic optimization, which we believe to be of independent interests.
\end{abstract}

\section{Introduction}
\label{sec:introduction}
Following the celebrated seminal work of Kempe et al.~\cite{kempe03journal}, the {\em influence maximization} (IM) problem has been extensively studied over past decades. 
Influence maximization is the problem of selecting at most $k$ {\em seed nodes} that maximizes the influence spread on a given social network and diffusion model.
It provides mathematical models for information diffusions and has numerous real world applications, such as viral markets, rumor controls, etc. 
In the past years, the IM problem has been studied under different context such as outbreak detection\cite{Leskovec2007celf}, topic-aware influence propagation~\cite{barbieri2012topic}, 
competitive and complementary influence maximization \cite{lu2015competition} etc., and both theoretical and practical efficient algorithms have been developed\cite{Chen2009efficient,chen2010sharpphard,borgs2014rrset,tang2014newrrset,tang2015influence}. 
See the recent survey~\cite{chen2013information,LiFWT18} for more detailed reference.

Meanwhile, stimulating by the real life demand, researchers in recent years begin to consider this classical problem in the adaptive setting. 
In the {\em adaptive influence maximization} problem, instead of consuming all budgets and selecting the seed set all at once, we are allowed to select seeds one after another, making future decisions based on the
propagation feedback gathered from the previous seeds.
Two feedback models are typically considered~\cite{GoloKrause11}: {\em myopic feedback}, where only the 
one-step propagation from the selected seed to its immediate out-neighbors are included in the feedback, and
{\em full-adoption feedback}, where the entire cascade from the seed is included in the feedback.
This adaptive decision process can potential bring huge benefits but it also brings technical challenges, since adaptive policies are usually hard to design and analyze, and the adaptive decision process can be slow in practice.
Thus, a crucial task in this area is to decide whether and how much adaptive policy is really superior over
the non-adaptive policy. 
The {\em adaptivity gap} quantifies to what extent adaptive policy outperforms a non-adaptive one and it is defined as the supremum ratio between the optimal adaptive policy and  the optimal non-adaptive policy.
The above question has been answered recently when only {myopic feedback} are available~\cite{peng2019adaptive,pmlr-v97-fujii19a} and constant upper bounds on adaptivity gap have been derived.

In this paper, we consider the influence maximization problem in the independent cascade (IC) model with
{\em full-adoption feedback}.
Even though the full-adoption feedback under the IC model satisfies an important property called
{\em adaptive submodularity}, the analysis of its adaptivity gap is more challenging because
the feedback obtained from different seed nodes are no longer independent --- feedback from one seed
contains multiple-step cascade results, and thus it may already contain partial feedback from another
seed.
Therefore, results from existing studies on the 
adaptivity gap of general classes of stochastic adaptive optimization problems \cite{asadpour2015maximizing,gupta2016algorithms,gupta2017adaptivity,bradac2019near} cannot be applied,
since they all rely on independent feedback assumption.

In this study, we are able to 
derive nontrivial constant upper bounds on several families of graphs, including in-arborescences, out-arborescence and bipartite graphs, which have been 
the targets of many studies in influence maximization (see Section~\ref{sec:related-work} for more details). 
Formally, we have 
(i) when the influence graph is an in-arborescence, the adaptivity gap is between $[\frac{e}{e - 1}, \frac{2e}{e - 1}]$ (Section~\ref{sec:in-arborescence} and Section~\ref{sec:adaptivity-gap-lower}), 
(ii) when the influence graph is an out-arborescence, the adaptivity gap is between $[\frac{e}{e - 1}, 2]$ (Section~\ref{sec:adaptivity-gap-out-arborescence} and Section~\ref{sec:adaptivity-gap-lower}) and 
(iii) the adaptivity gap for the bipartite influence graph is $\frac{e}{e - 1}$ (Section~\ref{sec:adaptivity-gap-bipartite}). 
Our upper bounds on arborescences are the first constant upper bounds in the full-adoption feedback model and our upper bound on bipartite graphs improves the results in~\cite{pmlr-v97-fujii19a, hatano2016adaptive}.

The main technical contributions in this paper are on the adaptivity gaps for arborescences, in which the feedback information can be correlated and all previous methods failed. We adopt two different proof strategies to overcome the difficulty of dependent feedback. For in-arborescences, we follow the framework in~\cite{asadpour2015maximizing} and construct a Poisson process to relate the influence spread of the optimal adaptive policy and the {\em multilinear extension}. The analyses are non-trivial due to the correlated feedback. We need to delicately decompose the marginal gain of the Poisson process and give upper bounds on each terms. The key observation we have for in-arborescences is that the {\em boundary} of the active nodes shrinks during the diffusion process. For out-arborescences, we again relate the influence spread of the multilinear extension to the optimal policy, but using a completely different proof strategy. The key observation for out-arborescences is that the predecessors of each node form a directed line thus proving a stronger results on this line is sufficient. We derive a family of constraints on the optimal adaptive policy and telescope the marginal gains of the multilinear extension, combining these two could yield our results.

Due to the space constraint, detailed proofs and some additional materials are moved into the appendix.


\vspace{-2mm}
\subsection{Related Work}
\label{sec:related-work}

A number of studies \cite{bharathi2007competitive, wang2016bharathi, lu2017solution, lin2017boosting} have focused on the influence maximization problem on arborescences and interesting theoretical results have been found with this special structural assumption. 
Bharathi et al. \cite{bharathi2007competitive} initiate the study on arborescences and derive polynomial-time approximation scheme (PTAS) for bidirected trees. 
For in-arborescences, Wang et al. \cite{wang2016bharathi} give a polynomial time algorithm in the linear threshold (LT) model and Lu et al. \cite{lu2017solution} prove NP hardness results under the independent cascade model. 

The influence maximization problem on one-directional bipartite graphs has been studied by~\cite{alon2012optimizing, soma2014optimal, hatano2016adaptive}, and it has applications on advertisement selections. Especially, Hatano et al. \cite{hatano2016adaptive} consider the problem in the adaptive setting and derive non-adaptive algorithms with theoretical guarantees. 

Initiated by the pioneering work of \cite{GoloKrause11}, a recent line of work~\cite{tong2017adaptive,yuan2017adaptive,salha2018adaptive,sun2018multi, pmlr-v97-fujii19a,peng2019adaptive} focus on the adaptive influence maximization problem and develop both theoretical results and practical methods. Golovn and Krause \cite{GoloKrause11} propose the novel concept of adaptive submodularity and applied it to the adaptive influence maximization problem. 
They prove that with full-adoption feedback in the IC model, the influence spread function satisfies the adaptive submodularity, thus a simple adaptive greedy algorithm could achieve the $(1 - 1/e)$ approximation ratio. Fujii et al. \cite{pmlr-v97-fujii19a} generalize the notion and propose weakly adaptive submodularity. They consider the adaptivity gap on both LT and IC models, when the influence graph is bipartite. While they prove a tight upper bound of 2 for the LT model, their bound for IC model depends on the structure of the graph and can be far worse than $1 - 1/e$. 
In contrast, in this paper we provide the tight bound of $1-1/e$ with a simple analysis in this case.
Recently, Peng and Chen \cite{peng2019adaptive} consider myopic feedback model and prove an upper bound of 4
for the adaptivity gap. 
Singer and his collaborators have done a series of studies on adaptive seeding and studied the adaptivity
gap in their setting \cite{seeman2013adaptive,singer2016influence,badanidiyuru2016locally}, but their model is a two-step adaptive model with the first step purely for referring
to the seed candidates, and thus their model is very different from adaptive influence maximization
of this paper and other related work above.

From theoretical side, there are two lines of works \cite{asadpour2015maximizing, adamczyk2016submodular, gupta2016algorithms, gupta2017adaptivity, bradac2019near} on the adaptivity gap that are most relevant to ours. Asadpour et al. \cite{asadpour2015maximizing} study the stochastic submodular optimization problem. 
They use multilinear extensions to transform an adaptive strategy to a non-adaptive strategy and give a tight upper bound of $\frac{e}{e - 1}$. 
Their methods inspire our work but they cannot be directly applied to our settings, since the feedback information are not independent in the full-adoption feedback model. We defer further discussion about the difference to Section~\ref{sec:in-arborescence}.
Another line of work \cite{gupta2016algorithms,gupta2017adaptivity,bradac2019near} focus on the stochastic probing problem. 
They transform any adaptive policy to a random walk non-adaptive policy and Bradac et al.~\cite{bradac2019near} finally prove a tight upper bound of 2 for prefix constraints. 
	\vspace{-2mm}
\section{Preliminaries}
\label{sec:preliminarly}

In this paper, we focus on the well known {\em independent cascade} (IC) model as the diffusion model. In the IC model, the social network is described by a directed influence graph $G = (V, E, p)$ ($|V| = n$), and there is a probability $p_{uv}$ associated with each edge $(u, v) \in E$. 
The \emph{live-edge} graph $L = (V, L(E))$ is a random subgraph of the influence graph $G$, 
	where each edge $(u, v)\in E$ appears in $L(E)$ independently with probability $p_{uv}$. If the edge appears in $L(E)$, we say it is \emph{live}, otherwise we say it is \emph{blocked}. We use $\aL$ to denote all possible live-edge graphs and $\cP$ to denote the probability distribution over $\aL$. The diffusion process can be described by the following discrete time process. At time $t = 0$, a seed set $S \subseteq V$ is activated and a live-edge graph $L$ is sampled from the probability distribution $\cP$ (i.e., each edge $(u, v)$ will be live with probability $p_{uv}$). At time $t = 1, 2, \cdots$, a node $u \in V$ is active if (i) $u$ is active at time $t - 1$ or (ii) one of $u$'s in-neighbor is active at time $t - 1$. The diffusion process will end when there are no new nodes been activated. We use $\Gamma(S, L)$ to denote the set of active nodes at the end of diffusion, or equivalently, the set of nodes reachable from set $S$ under live-edge graph $L$.  We define the influence reach function $f:\{0,1\}^{V} \times \aL \rightarrow \R^{+}$ as $f(S, L) := |\Gamma(S, L)|$. 
	Then the influence spread of a set $S$, denoted as $\fa(S)$, is defined as the expected number of active nodes at the end of the diffusion process, i.e., $\fa(S):= \E_{L\sim \cP}[f(S, L)]$. 

We formally state the (non-adaptive) influence maximization problem as follow.
\vspace{-2mm}
\begin{Def}[Non-adaptive influence maximization]
	\label{def:non-adaptive-im}
	The non-adaptive influence maximization (IM) problem is the problem of given an influence graph $G=(L, V, p)$ and a budget $k$, finding a seed set $S^{\star}$ of size at most $k$ that maximizes the influence spread, i.e., find $S^{\star} = \argmax_{S\subseteq V, |S| \leq k} \fa(S)$.
\end{Def}

In the adaptive setting, instead of committing the entire seed set all at once, we are allowed to select the seed node one by one. After we seed a node, we can get some feedback about the diffusion state from the node. Formally, a {\em realization} $\phi$ is a function $\phi: V \rightarrow O$, mapping a node $u$ to its state, i.e., the feedback we obtain when we select the node $u$ as a seed. 
The realization $\phi$ determines the status of all edges in the influence graph and it is one-to-one correspondence to a live-edge graph. Henceforth, in the rest of the paper, we would use $\phi$ to refer to both the realization and the live-edge graph interchangeably.  The feedback information depends on the feedback model and in this paper, we consider the {\em full-adoption} feedback model. In the full-adoption feedback model, after we select a node $u$, we get to see the status of all out-going edges of nodes $v$ that are reachable from $u$ in the live-edge graph. 
In another word, we get to see the full cascade of the node $u$. At each step of the adaptive seeding process, our observation so far is represented by a {\em partial realization} $\psi \subseteq V\times O$,
which is a collection of nodes and states, $(u, \phi(u))$, we have observed so far. We use $\dom(\psi)$ to denote the set $\{u : u \in V, \exists (u, o) \in \psi\}$, that is, all nodes we have selected so far. 
For two partial realizations $\psi$ and $\psi'$, we say $\psi$ is a sub-realization of 
	$\psi'$ if $\psi \subseteq \psi'$  when treating $\psi$ and
	$\psi'$ as subsets of $V\times O$.

An adaptive policy $\pi$ is a mapping from partial realizations to nodes. Given a partial realization $\psi$, we use $\pi(\psi)$ to represent the next seed selected by $\pi$. After selecting node $\pi(\psi)$, our observation (partial realization) grows as $\psi' = \psi \cup (\pi(\psi), \phi(\pi(\psi)))$ and the policy $\pi$ would pick the next node based on the new partial realization $\psi'$. Given a realization $\phi$, we use $V(\pi, \phi)$ to denote the seed set selected by the policy $\pi$. 
The {\em adaptive influence spread} of the policy $\pi$ is defined as the expected number of active nodes under the policy $\pi$, i.e., $\fa(\pi) := \E_{\Phi\sim \cP}[f(V(\pi, \Phi), \Phi)]$. We define $\Pi(k)$ as the set of policies $\pi$, such that for any possible realization $\phi$, $|V(\pi, \phi)| \leq k$. The adaptive influence maximization problem is formally stated as follow.
\begin{Def}[Adaptive influence maximization]
	\label{adaptive-influence-maxmization}
	The adaptive influence maximization (AIM) problem is the problem of given an influence graph $G=(L, V, p)$ and a budget $k$, finding a feasible policy $\pi \in \Pi(k)$ that maximizes the adaptive influence spread, i.e., find $\pi^{\star} = \argmax_{\pi \in \Pi(k)} \fa(\pi)$.
\end{Def}

In this paper, we study the {\em adaptivity gap} of the influence maximization problem under full-adoption feedback model. The adaptivity gap measures the supremacy of the optimal adaptive policy over the optimal non-adaptive policy.
We use $\OPT_{N}(G, k)$ (resp. $\OPT_{A}(G, k)$) to denote the influence spread of the optimal non-adaptive (resp. adaptive) policy for the IM problem on the influence graph $G$ with a budget $k$. 
\begin{Def}[Adaptivity gap]
	\label{def:adaptivity-gap}
The adaptivity gap for the IM problem in the IC model with full-adoption feedback is defined as
	\begin{align}
	\label{eq:adaptivity-gap-def}
	\sup_{G, k}\frac{\OPT_{A}(G, k)}{\OPT_{N}(G, k)}.
	\end{align}
\end{Def}


We prove constant upper bounds on the adaptivity gap for several classes of graphs, including the {\em in-arborescence} and the {\em out-arborescence}.

\begin{Def}[In-arborescence]
	\label{def:in-arborescence}
	We say an influence graph $G=(V, E, p)$ is an in-arborescence when the underline graph is a directed tree with a root $u$, such that for any node $v\in V$, the unique path between nodes $u$ and $v$ is directed from $v$ to $u$. In other words, the information propagates from leaves to the root.
\end{Def}
\begin{Def}[Out-arborescence]
	\label{def:out-arborescence}
	We say an influence graph $G=(V, E, p)$ is an out-arborescence when the underline graph is a directed tree with a root $u$, such that for any node $v\in V$, the unique path between nodes $u$ and $v$ is directed from $u$ to $v$. In other words, the information propagates from the root to leaves.
\end{Def}
A set function $f: V \rightarrow \R^{+}$ is said to be submodular if for any set $A \subseteq B \subseteq V$ and any element $u \in V\backslash B$, $\Delta_f(u | A):= f(A \cup \{u\}) - f(A) \geq f(B \cup\{u\}) - f(B) = \Delta_f(u | B)$. We call $\Delta_f(u | A)$ the marginal gain for adding element $u$ to the set $A$.
Moreover, the function $f$ is said to be monotone if $f(B)\geq f(A)$.  
Under the IC model, the influence spread function $\fa(\cdot)$ is proved to be {\em submodular} and {\em monotone} \cite{Kempe2003maximizing}, thus given value oracles for $\fa(\cdot)$, the greedy algorithm is $1 - 1/e$ approximate to the optimal non-adaptive solution. 

In the adaptive submodular optimization scenario, a similar notion corresponds to the submodularity is called the {\em adaptive submodularity}. For a function $f$ and a partial realization $\psi$,  we write $\Phi \sim \psi$ to say that the realization $\Phi$ is {\em consistent} with the partial realization $\psi$, i.e., $\Phi(u) = \psi(u)$ for any $u \in \dom(\psi)$, then the conditional marginal gain for an element $u \notin \dom(\psi)$ is defined as 
$\Delta_{f, \cP}(u | \psi):=\E_{\Phi\sim \cP}\left[f\left(\dom(\psi)\cup \{u\}\right) - f\left(\dom(\psi)\right) | \Phi \sim \psi \right],$. A function $f$ is said to be adaptive submodular with respect to $\cP$ if for any partial realizations $\psi \subseteq \psi'$ and any element $u \in V\backslash \dom(\psi')$, $\Delta_{f, \cP}(u | \psi) \geq \Delta_{f, \cP}(u | \psi')$. Moreover, the function $f$ is {\em adaptive monotone} with respect to $\cP$ if $\Delta_{f, \cP}(u | \psi)\geq 0$ for any feasible partial realization $\psi$. Golovin and Krause~\cite{GoloKrause11} shows the following important result, which will
	be used in our analysis.
\begin{prop}[\cite{GoloKrause11}]
\label{prop:infAdaptiveSubmodular}
Influence reach function $f$ is adaptive submodular and adaptive monotone 
	with respect to the live-edge graph distribution $\cP$ 
	under the independent cascade model with full-adoption feedback.
\end{prop}

\vspace{-2mm}
The following two definitions are very important to our later analysis.

\begin{Def}[Multilinear extension]
	\label{definition:multilinear-extension}
	The multilinear extension $F: [0,1]^{V} \rightarrow \R^{+}$ of the influence spread function $\fa$ is defined as
	\begin{align}
	\label{eq:multilinear-extension-def}
	F(x_1, \cdots, x_n) = \sum_{S \subseteq V}\left[ \left(\prod_{i \in S}x_i\prod_{i \notin S}\left(1 - x_i\right)\right) \fa(S)\right].
	\end{align}
\end{Def}
We remark that the multilinear extension $F(\cdot)$ is monotone and DR-submodular \cite{kempe03journal}, when the original function $\fa(\cdot)$ is monotone and submodular. 
A vector function $f$ is DR-submodular if for any two vectors $(x_1, \ldots, x_n) \le (y_1, \ldots, y_n)$
	(coordinate-wise), for any $\delta > 0$, any $j \in [n]$, 
	$f(x_1, \ldots, x_j + \delta, \ldots, x_n) - f(x_1, \ldots, x_n) \ge 
	f(y_1, \ldots, y_j + \delta, \ldots, y_n) - f(y_1, \ldots, y_n)$.
For any configuration $(x_1, \cdots, x_n)$, we use $\fps(x_1, \cdots, x_n)$ to denote the optimal adaptive strategy under this configuration. Formally,

\begin{Def}[Adaptive influence spread function based on an optimal adaptive policy]
	\label{definition:optimal-adaptive-strategy}
	We define $\fps: [0,1]^{V} \rightarrow \R^{+}$ as:
	\begin{align}
	\label{eq:optimal-adaptive-strategy-def}
	\fps(x_1, \cdots, x_n) = \sup_{\pi}\left\{\fa(\pi) : \Pr_{\Phi\sim\cP}\left[i \in V(\pi, \Phi)\right] = x_i, \, \forall i \in [n] \right\}.
	\end{align}
\end{Def}

\section{Adaptivity Gap for In-arborescence}
\vspace{-2mm}
\label{sec:in-arborescence}
In this section, we give an upper bound on the adaptivity gap when the influence graph is an in-arborescence,
	as stated in the following theorem. 
\begin{thm}
	\label{thm:adaptivity-gap-in-arborescence}
	When the underline influence graph is an in-arborescence, the adaptivity gap for the IM problem in the IC model with full adoption feedback is at most $\frac{2e}{e - 1}$.
\end{thm}

Our approach follows the general framework of \cite{asadpour2015maximizing}, i.e., we use the multilinear extension to transform an adaptive policy to a non-adaptive policy, and construct a Poisson process to connect the influence spread of the non-adaptive policy to the adaptive policy. 
Once we have done this, combining with the rounding procedure in \cite{calinescu2011maximizing, chekuri2010dependent}, we can derive an upper bound on the adaptive policy.
However, remembering that the main difficulty of our problem comes from the correlation of the feedback, directly applying the analyses in \cite{asadpour2015maximizing} does not work. 
Our methods have several key differences comparing to \cite{asadpour2015maximizing}.
To be more specific, we can no longer directly relate the dynamic marginal gain of the Poisson process to the influence spread of the adaptive policy. 
Instead, we need to delicately decompose the marginal gain into two parts (see Lemma \ref{lem:possion-process-vs-adapt-opt} and Lemma \ref{lem:two-hop-influence-upper-bound}). 
The first part can be related to the optimal adaptive strategy via the adaptive submodularity, while the second part can be related to a (randomize) non-adaptive policy. 
However, this non-adaptive policy is not guaranteed to be bounded by $\OPT_{N}(G, k)$ (the optimal non-adaptive policy of budget $k$), 
because the size of the seed set is random and can potentially be very large. 
We utilize the {\em``weak concavity''} of the optimal solution to show that it is enough to consider the expected size of the (random) seed set,
and then we give an upper bound on this expected size for an in-arborescence. 
This upper bound relies on a crucial property of the in-arborescence, i.e., the {\em boundary} (see Definition~\ref{definition:boundary-partial-realization}) size always shrinks during the diffusion process of the information (see Lemma~\ref{lem:in-arborescence-boundary}).
Putting things together, we get a differential inequation that relates the dynamic marginal gain of the Poision process to both an optimal adaptive policy and the optimal non-adaptive policy. Solving the differential inequation yields a lower bound on the multilinear extension and it gives an upper bound on the adaptivity gap. 
We remark that one noticeable difference of our bound on the adaptivity gap is that it does not hold for the matroid constraint (which holds in \cite{asadpour2015maximizing}), even though the multilinear extension was original designated to handle matroid constraints.

Following the work \cite{asadpour2015maximizing, vondrak2007submodularity}, for any configuration $(x_1, \cdots, x_n)$, we consider the following Poisson process, which will indirectly relate the multilinear extension $F(x_1, \cdots, x_n)$ to the optimal adaptive strategy $\fps(x_1, \cdots, x_n)$. 

\noindent{\bf Poisson Process.\ \ }
There are $n$ independent Poisson clocks $C_1, \cdots, C_n$, the clock $C_i$ ($i \in [n]$) sends signals with rate $x_i$. Whenever a clock $C_i$ sends out a signal, we select node $i$ as a seed and gather feedback $\phi(i)$ according to the underline realization $\phi$. We use $\Psi(t)$ to denote the partial realization at time $t$ and we start with $\Psi(0)$ as $\emptyset$. 
We note that $\Psi(t)$ is a random partial realization that contains
	(a) random time points $t_i \le t$ at which clock $C_i$ sends a signal; and
	(b) for each $t_i \le t$, the feedback $\phi(i)$ of seed node $i$ based on the corresponding 
	live-edge graph. 
The Poisson process end at $t = 1$.
Note that the Poisson process is parameterized by $(x_1, x_2, \ldots, x_n)$, but we ignore these parameters
	in the notation $\Psi(t)$.

With a slight abuse of notation, we define $\Gamma(\psi)$ as the set of nodes reachable from $\dom(\psi)$ under partial realization $\psi$ and we define $f(\psi) = |\Gamma(\psi)|$. Notice that in the full-adoption feedback model, a partial realization could already determine all nodes reachable from the seed set.
The following lemma states that at the end of the Poisson process, i.e., when $t = 1$, the expected influence spread of $\E[f(\Psi(1))]$ is no greater than the influence spread of $F(x_1, \cdots, x_n)$.
\vspace{-2mm}
\begin{restatable}{lem}{lemPoissonProcessMultilinearExtension}
\label{lem:possion-process-vs-multilinear-extension}
$\E\left[f\left(\Psi(1\right))\right] = F\left(1 - e^{-x_1}, \cdots, 1 - e^{-x_n}\right) \leq F(x_1, \cdots, x_n)$.
\end{restatable}
\vspace{-2mm}

We then give a lower bound on the dynamic marginal gain of the influence spread in the Poisson process.

\begin{restatable}{lem}{lemPoissonProcessAdaptiveOptimal}
\label{lem:possion-process-vs-adapt-opt}
For any $t \in [0, 1]$ and any fixed partial realization $\psi$, we have
\begin{align}
\label{eq:in-arborescence-eq1}
\E\left[\frac{d f\left(\Psi(t)\right)}{dt} \mid \Psi(t) = \psi\right] \geq \fps\left(x_1, \cdots, x_n\right) - \fa\left(\Gamma(\psi)\right).
\end{align}
\end{restatable}

Next, we introduce the concept of the {\em boundary} of a partial realization.
\begin{Def}[Boundary of a partial realization]
	\label{definition:boundary-partial-realization}
	In the full-adoption feedback model, for any partial realization $\psi$, we use $\partial(\psi)$ to denote the boundary of the partial realization, i.e., the set of nodes $\partial(\psi) \subseteq \Gamma(\psi)$ 
	with minimum cardinality such that there is no directed edges in the original graph $G$ from $\Gamma(\psi) \backslash \partial(\psi)$ to $V\backslash \Gamma(\psi)$. We remark that when there are more than one such sets, we take an arbitrary one.
\end{Def}
\vspace{-2mm}

The main property we rely on the structure of an in-arborescence is that the boundary of any partial realization can be bounded by the number of seeds that have been selected. Formally, we have
\begin{restatable}{lem}{lemInArborescenceBoundary}
	\label{lem:in-arborescence-boundary}
	When the influence graph is an in-arborescence, for any partial realization $\psi$, we have $|\partial(\psi)| \leq  |\dom(\psi)|$.
\end{restatable}
Now, we give an upper bound on the term $\fa(\Gamma(\psi))$.
\begin{restatable}{lem}{lemInfluenceSpreadTwoHop}
	\label{lem:two-hop-influence-upper-bound}
	For any partial realization $\psi$
	\begin{align}
	\label{eq:in-arborescence-eq5}
	\fa(\Gamma(\psi)) \leq |\Gamma(\psi)| + \fa(\partial(\psi)).
	\end{align}
	Moreover, when the influence graph is an in-arborescence, we have
	\begin{align}
	\label{eq:in-arborescence-eq6}
	\fa(\Gamma(\psi)) \leq |\Gamma(\psi)| + \OPT_{N}(G, |\dom(\psi)|).
	\end{align}
\end{restatable}

For any fixed influence graph $G$, we can view $\OPT_{N}(G, k)$ as a function of the budget $k$, we prove that $\OPT_{N}(G, k)$ is {\em``weak concave''} for $k$, as stated in the following lemma.
\begin{restatable}{lem}{lemNearConcaveOpt}
	\label{lem:near-concave-opt}
	For any fixed influence graph $G$, let $X$ be a random variable taking value from $\{0, 1\cdots, n\}$, with mean value $\E[X] = k$. Then we have
	\begin{align}
	\label{eq:near-concave-opte}
	\E\left[\OPT_N(G, X)\right] \leq \frac{e}{e - 1}\OPT_{N}\left(G, \E[X]\right) = \frac{e}{e - 1}\OPT_N(G, k).
	\end{align}
\end{restatable}

Putting things together, we can prove Theorem \ref{thm:adaptivity-gap-in-arborescence}.

\begin{proof}[Proof of Theorem \ref{thm:adaptivity-gap-in-arborescence}]
	When the influence graph is an in-arborescence, for any configuration ($x_1, \cdots, x_n$) satisfying $\sum_{i}x_i = k$, for any $t\in [0,1]$ and any fixed partial realization $\psi$, we have
	\begin{align}
	\E\left[\frac{d f\left(\Psi(t)\right)}{dt} | \Psi(t) = \psi\right] &\geq \fps\left(x_1, \cdots, x_n\right) - \fa\left(\Gamma(\psi)\right)&\text{by Lemma \ref{lem:possion-process-vs-adapt-opt}} \notag\\ 
	&\geq \fps(x_1, \cdots, x_n) - |\Gamma(\psi)| - \fa(\partial(\psi)) &\text{by Lemma \ref{lem:two-hop-influence-upper-bound}} \notag\\
	&\geq \fps(x_1, \cdots, x_n) - |\Gamma(\psi)|  - \OPT_N(G, |\dom(\psi)|)&\text{by Lemma \ref{lem:in-arborescence-boundary}} \notag\\
	&= \fps(x_1, \cdots, x_n) - f(\Psi(t)) - \OPT_N(G, |\dom(\Psi(t))|). \label{eq:in-arborescence-eq10}
	\end{align}
	Taking expectation over $\Psi(t)$, we have for any $t\in [0,1]$, 
	\begin{align}
	\frac{d}{dt}\E\left[f\left(\Psi(t)\right)\right] &= \E\left[\frac{d f\left(\Psi(t)\right)}{dt}\right]
		= \E_{\Psi(t)}\E\left[\frac{d f\left(\Psi(t)\right)}{dt} \mid \Psi(t) \right]
		\notag \\
	&\geq \fps(x_1, \cdots, x_n) - \E\left[f(\Psi(t))\right] - \E\left[\OPT_N(G, |\dom(\Psi(t))|)\right]\notag\\ 
	&\geq \fps(x_1, \cdots, x_n) - \frac{e}{e - 1}\OPT_{N}(G, k) - \E\left[f(\Psi(t))\right].\label{eq:in-arborescence-eq11}
	\end{align}
	The first equality above is by the linearity of expectation.
	The second equality above is by the law of total expectation.
	The first inequality is by Eq.\eqref{eq:in-arborescence-eq10}, and
		the second inequality holds due to Lemma \ref{lem:near-concave-opt} and the fact that
	\begin{align*}
	\label{eq:in-arborescence-eq12}
	\E\left[\dom(\Psi(t))\right] =\sum_{i} (1 - e^{-tx_i}) \leq \sum_{i}tx_i \leq k
	\end{align*} 
	for any $t \leq 1$. Solving the above differential inequality in Eq.\eqref{eq:in-arborescence-eq11} gives us
	\begin{align}
	\E\left[f(\Psi(t))\right] \geq (1 - e^{-t})\left[\fps(x_1, \cdots, x_n) - \frac{e}{e - 1}\OPT_{N}(G, k)\right].
	\end{align}
	In particular, when $t = 1$, we have
	\begin{align}
	\E\left[f(\Psi(1))\right] \geq \left(1 - \frac{1}{e}\right)\left[\fps(x_1, \cdots, x_n) - \frac{e}{e - 1}\OPT_{N}(G, k)\right].\label{eq:in-arborescence-eq13}
	\end{align}
	Finally, we have
	\begin{align}
	\OPT_{N}(G, k) &= \sup_{x_1 + \cdots + x_n = k}F(x_1, \cdots, x_n)\notag\\
	& \ge \sup_{x_1 + \cdots + x_n = k} \E\left[f(\Psi(1))\right]  \notag \\
	&\geq \sup_{x_1 + \cdots + x_n = k}\left(1 - \frac{1}{e}\right)\left[\fps(x_1, \cdots, x_n) - \frac{e}{e - 1}\OPT_{N}(k)\right]\notag\\
	&\ge \left(1 - \frac{1}{e}\right)\OPT_{A}(G, k) - \OPT_{N}(G,k). \label{eq:in-arborescence-eq14}
	\end{align}
	The first equality above comes from the pipage rounding procedure in \cite{calinescu2011maximizing}. 
	The first inequality above is by Lemma~\ref{lem:possion-process-vs-multilinear-extension}.
	The second inequality is by Eq.~\eqref{eq:in-arborescence-eq13}.
	The last equality is by the definition of $f^+$ (Definition~\ref{definition:optimal-adaptive-strategy}).
	Thus we conclude that the adaptivity gap is at most $\frac{2e}{e - 1}$ in the case of an in-arborescence.
\end{proof}

	\section{Adaptivity Gap for Out-arborescence}
\label{sec:adaptivity-gap-out-arborescence}
\vspace{-2mm}

In this section, we give an upper bound on the adaptivity gap when the influence graph is an out-arborescence. Formally,
\begin{thm}
	\label{thm:adaptivity-gap-out-arborescence}
	When the influence graph is an out-arborescence, the adaptivity gap for the IM problem in the IC model with full-adoption feedback is at most 2.
\end{thm}
We first introduce some notations. For any node $u \in V$ and any seed set $S \subseteq V$, we define $\fa_{u}(S):= \Pr_{\Phi}\left[u \in \Gamma(S, \Phi)\right]$, i.e., the probability that the node $u$ 
	activated when $S$ is the seed set. 
Similarly, for any adaptive policy $\pi$, we define $\fa_{u}(\pi) :=\Pr_{\Phi}\left[u \in \Gamma(V(\pi, \Phi), \Phi)\right]$, i.e., the probability that the node $u$ is activated under policy $\pi$. We would extend the definition for the multilinear extension (Definition~\ref{definition:multilinear-extension}) and the definition for $\fps$ (Definition~\ref{definition:optimal-adaptive-strategy}) correspondingly. To be more specific, we define
\begin{align}
\label{eq:multilinear-extension-per-node-def}
F_{u}(x_1, \cdots, x_n) = \sum_{S \subseteq [n]}\left[ \left(\prod_{i \in S}x_i\prod_{i \notin S}\left(1 - x_i\right)\right) \fa_{u}(S)\right],
\end{align}
and 
\begin{align}
\label{eq:optimal-adaptive-strategy-per-node-def}
\fps_{u}(x_1, \cdots, x_n) = \sup_{\pi}\left\{\fa_{u}(\pi) : \Pr_{\Phi\sim\cP}\left[i \in V(\pi, \Phi)\right] = x_i, \, \forall i \in [n] \right\}.
\end{align}
In order to show Theorem~\ref{thm:adaptivity-gap-out-arborescence}, we again transform an adaptive policy to a non-adaptive policy and compare their influence spread. Here, we utilize a new approach based on the structure of out-arborescences and prove a stronger result. That is, we would prove that the probability for any node $u$ become active in the multilinear extension (policy) is at least half of the optimal adaptive policy (see Eq.~\eqref{eq:out-arborescence-eq1}). This requires use to give fine-grained bound on the optimal adaptive policy (Lemma~\ref{lemma:upper-bound-adaptive}) and the multilinear extension (Lemma~\ref{lem:induction-multilinear-ex}).
\begin{proof}
	When the influence graph is an out-arborescence, for any node $u\in V$ and any configuration ($x_1, \cdots, x_n$), we are going to prove that
	\begin{align}
	\label{eq:out-arborescence-eq1}
	\fps_{u}(x_1, \cdots, x_n) \leq 2 F_{u}(x_1, \cdots, x_n).
	\end{align}
	This suffices to prove Theorem \ref{thm:adaptivity-gap-out-arborescence} because
	\begin{align*}
&\OPT_{N}(G, k) = \sup_{x_1 + \cdots + x_n = k}F(x_1, \cdots, x_n)  
=  \sup_{x_1 + \cdots + x_n = k} \sum_{u}F_{u}(x_1, \cdots, x_n) \\
&\ge \frac{1}{2}  \cdot \sup_{x_1 + \cdots + x_n = k}  \sum_{u} \fps_{u}(x_1, \cdots, x_n) \ge \frac{1}{2}  \cdot  \sup_{x_1 + \cdots + x_n = k} \fps(x_1, \cdots, x_n)  
\ge \frac{1}{2}  \cdot \OPT_{A}(G, k).
\end{align*}
	We note that node $u$'s predecessors (nodes that can reach node $u$ in the original graph) 
	form a directed line when the influence graph is an out-arborescence. 
	We slightly abuse the
		notation and use node $i$ to indicate the  $(i - 1)^{th}$ predecessor of node $u$, notice that node $u$ itself is represented as node 1. 
		We ignore all other nodes since they do not affect either sides of Eq.~\eqref{eq:out-arborescence-eq1}. We use $p_i$ to denote the probability that the node $i$ can reach node $1$. The following lemma gives an upper bound on the optimal adaptive strategy.
	\begin{restatable}{lem}{lemUpperAdaptiveOptimal}
		\label{lemma:upper-bound-adaptive}
		For any $i $, $\fps_{1}(x_1, \cdots, x_n) \leq \sum_{j = 1}^{i} x_j p_j + p_{i + 1}.$ 
	\end{restatable}
	
	We measure the marginal contribution of node $i$ in the next lemma.
	Intuitively, we can see that $F_{1}(0, \cdots, 0, x_i, \cdots, x_n) -  F_{1}(0, \cdots, 0, x_{i +1}, \cdots, x_n)$
	measures the marginal contribution of $i$ in activating node $1$, when node $i$ moves from
	no probability of being selected as a seed to probability of $x_i$ being selected as the seed, 
	under the situation that no nodes in $\{1, \ldots, i-1 \}$ can be seeds while node $j > i$ has probability
	$x_j$ being selected as a seed.
	Then this marginal contribution only happens when all three conditions hold:
		(a) possible seeds in $\{i+1, \ldots, n\}$ cannot activate $i$, which has probability
			$1 - F_i(0, \cdots, 0, x_{i +1}, \cdots, x_n)$; 
		(b) node $i$ is activated as a seed, which has probability $x_i$, and
		(c) node $i$ passes influence and activate node $1$, which has probability $p_i$.
	\begin{restatable}{lem}{lemMultilinearExtensionTelescope}
		\label{lem:induction-multilinear-ex}
		For any $i$, we have
		\begin{align*}
		F_{1}(0, \cdots, 0, x_i, \cdots, x_n) -  F_{1}(0, \cdots, 0, x_{i +1}, \cdots, x_n) = x_i p_i \left(1 - F_i(0, \cdots, 0, x_{i +1}, \cdots, x_n)\right).
		\end{align*}
	\end{restatable}
	
	Back to the proof of Theorem \ref{thm:adaptivity-gap-out-arborescence},
	we use $j$ to denote the minimum index that satisfies $F_{j}(0, \cdots, x_{j +1}, \cdots, x_n) > \frac{1}{2}$.
	If such index does not exist, we simply set $j =  n + 1$. Then, we have
	\begin{align}
	&F_{1}(x_1, \cdots, x_n) \notag\\
	=& \sum_{i = 1}^{j - 1}\left(F_{1}(0, \cdots, 0, x_i, \cdots, x_n) -  F_{1}(0, \cdots, 0, x_{i +1}, \cdots, x_n)\right) + F_{1}(0, \cdots, 0, x_j, \cdots, x_n)\notag\\
	=&  \sum_{i = 1}^{j - 1} x_i p_i \left(1 - F_i(0, \cdots, 0, x_{i +1}, \cdots, x_n)\right) + F_{1}(0, \cdots, 0, x_j, \cdots, x_n)\notag\\
	=&  \sum_{i = 1}^{j - 1} x_i p_i \left(1 - F_i(0, \cdots, 0, x_{i +1}, \cdots, x_n)\right) + p_{j}\cdot F_{j}(0, \cdots, 0, x_j, \cdots, x_n)\notag\\	
	\geq& \frac{1}{2}\sum_{i = 1}^{j - 1}x_i p_i + \frac{1}{2} p_j\notag\\
	\geq& \frac{1}{2}\fps_{1}(x_1, \cdots, x_n).	\label{eq:out-arborescence-eq5}
	\end{align}
The second equality comes from Lemma \ref{lem:induction-multilinear-ex} and the last inequality comes from Lemma \ref{lemma:upper-bound-adaptive}. 
\end{proof}



	\section{Adaptivity Gap for One-Directional Bipartite Graphs}
\label{sec:adaptivity-gap-bipartite}
In this section, we give an upper bound on the adaptivity gap of the influence maximization problem in the IC model with full-adoption feedback under one-directional bipartite graphs $G(L, R, E, p)$, where
$L$ and $R$ are the two set of nodes on the left side and right side respectively, and 
$E \subseteq L \times R$ are a set of edges only pointing from a left-side node to a right-side node, 
and $p$ maps each edge to a probability.
Our upper bound is tight as it matches the lower bound derived in \cite{peng2019adaptive} and it also improves the results developed in \cite{pmlr-v97-fujii19a, hatano2016adaptive}. 
The proof strategy adopted for bipartite graphs is a relative easy application of our approaches in previous sections, it again relates the multilinear extension and the optimal strategy. 

\begin{restatable}{thm}{thmAdaptivityGapUpperBipartite}
	\label{thm:adaptivity-gap-upper-bipartite}
	When the influence graph is a one-directional bipartite graph $G(L, R, E, p)$, the adaptivity gap on the influence maximization problem in the IC model with full-adoption feedback is $\frac{e}{e - 1}$.
\end{restatable}
\vspace{-2mm}

\section{Lower Bounds on the Adaptivity Gap}
\label{sec:adaptivity-gap-lower}
In this section, we give an example showing that the adaptivity gap is no less than $e/(e-1)$ in the full-adoption feedback model, even when the influence graph is a directed line, a special case of both the in-arborescence and the out-arborescence.

\begin{thm}
\label{theorem:adaptivity-gap-lower-bound}
The adaptivity gap for the IM problem in the IC model with full-adoption feedback is at least $e/(e-1)$, even when the influence graph is a directed line.
\end{thm}
\begin{proof}
Consider the following influence graph $G(V, E, p)$: the graph is a directed line with vertex $v_{11}, \cdots, v_{1t}, v_{21}, \cdots, v_{2t}, \cdots,v_{k1}, \cdots, v_{kt}$, and each edge is live with probability $1 - 1/t$. Moreover, we have a budget $k$. Combining the following two claims, we can conclude that the adaptivity gap is greater than $e/(e-1)$.
\begin{restatable}{clm}{clmDirectLineAdaptive}
	\label{clm:full-adap-opt-exp}
	For any $\epsilon > 0$, if $k \geq 8/\epsilon^3$, we have $\E[\OPT_A(G, k)] \geq (1 - \epsilon)kt$. 
\end{restatable}
\vspace{-2mm}
\begin{restatable}{clm}{clmDirectLineNonAdaptive}
	\label{clm:full-nonadap-opt-exp}
	The optimal non-adaptive strategy is to select $v_{11}, \cdots, v_{k1}$ as seeds. Thus, we have $\E[\OPT_{N}(G, k)] = (1 - (1 -1/t)^t)kt$.\qedhere
\end{restatable}
\end{proof}
\vspace{-2mm}
\noindent{\bf Discussion on Existing Approaches.\ \ }
There are two types of strategies for proving upper bounds on adaptivity gaps. 
One common strategy is to convert any adaptive strategy to the multilinear extension as in \cite{seeman2013adaptive, asadpour2015maximizing} and our paper. 
The other is to convert the adaptive strategy to the {\em random walk non-adaptive strategy}~\cite{gupta2016algorithms,gupta2017adaptivity,bradac2019near}. 
Here we claim that using the instance constructed in Theorem~\ref{theorem:adaptivity-gap-lower-bound}, we can show that these two strategies can not yield better-than-$2$ upper bounds on the adaptivity gap. We defer the detailed discussions to the appendix.

\vspace{-3mm}
\section{Conclusion}
\label{sec:conclusion}
\vspace{-2mm}
In this paper, we consider several families of influence graphs and give the first constant upper bounds on adaptivity gaps for them under the full-adoption feedback model. Our methods tackle with the correlations on the feedback and hopefully can be applied to other adaptive stochastic optimization problems. For future directions, there are still gaps between our lower and upper bounds for both in-arborescences and out-arborescences, so it would be interesting to close the gap. Another open question is to settle down the adaptivity gap for general influence graphs under the full-adoption feedback model.
 


	\bibliographystyle{acmsmall}
	\bibliography{bibdatabase}
	\clearpage
	\section*{Appendix}
	\appendix
	For convenience, we restate the lemmas and theorems in the appendix before the proofs.

\section{Missing Proofs from Section~\ref{sec:in-arborescence}}

\lemPoissonProcessMultilinearExtension*
\begin{proof}
	Notice that in the Poisson process, the selection of seeds are actually independent of the realization of the influence graph. Moreover, seed nodes are selected independently. At the end of the process (when $t=1$), the node $i$ is selected as a seed with probability $1 - e^{-x_i}$. Thus we have
	\begin{align}
	\E\left[f\left(\Psi(1)\right)\right] &=\sum_{S \subseteq V}\left[ \left(\prod_{i \in S}(1 - e^{-x_i})\right)\prod_{i \notin S}\left(e^{-x_i}\right) \fa(S) \right]\notag\\
	&= F(1 - e^{-x_1}, \cdots, 1 - e^{-x_n}) \leq F(x_1, \cdots, x_n)\label{eq:poission-vs-multilinear-extension}.
	\end{align}
	The inequality holds due to the monotonicity of the multilinear extension $F(\cdot)$ and the fact that $1 - e^{-x} \leq x$.
\end{proof}

\lemPoissonProcessAdaptiveOptimal*

\begin{proof}
	First, we consider the left-hand side of Eq.~\eqref{eq:in-arborescence-eq1}. For any $t \in [0, 1], i \in [n]$ and small enough amount of time $dt$, the clock $C_i$ sends out signals with probability $x_i dt$ during the time interval $[t, t + dt]$. Since signals are sent out independently, the probability that more than one clock send out signals simultaneously in time interval $[t, t + dt]$ is of $O((dt)^2)$, which can be consider negligible comparing to $dt$. Thus we have
	\begin{align}
	\label{eq:in-arborescence-eq2}
	\E\left[f\left(\Psi(t + dt)\right) - f\left(\Psi(t)\right) | \Psi(t) = \psi\right] = \sum_{i \notin \dom(\psi)}x_i dt \cdot \Delta_{f}(i | \psi),
	\end{align}
	Rewriting the above equation, we derive that 
	\begin{align}
	\label{eq:in-arborescence-eq4}
	\E\left[\frac{df\left(\Psi(t)\right)}{dt} | \Psi(t) = \psi\right] =\sum_{i \notin \dom(\psi)}x_i \Delta_{f}(i | \psi) =\sum_{i \notin \Gamma(\psi)}x_i \Delta_{f}(i | \psi).
	\end{align}
	The second equality holds because $\Delta_{f}(i | \psi) = 0$ for any node $i \in \Gamma(\psi)$ in the full-adoption feedback model.
	Next, we consider the right-hand side of Eq.~\eqref{eq:in-arborescence-eq1}. We write $\bx = (x_1, \cdots, x_n)$ and use the indicator vector $\bI_{S} \in \{0, 1\}^{n}$ to denote an $n$-dimensional 0-1 vector, such that the coordinate $i$ is 1 if and only if $i \in S$. By the monotonicity of the function $\fps(\cdot)$, we have
	\begin{align}
	\label{eq:in-arborescence-eq3}
	\fps\left(x_1, \cdots, x_n\right) &\leq \fps\left(\bx \vee \bI_{\Gamma(\psi)}\right).
	\end{align} 
	Consider the optimal adaptive policy $\pi^+$ of $\fps(\bx \vee \bI_{\Gamma(\psi)})$
	as defined in Definition~\ref{definition:optimal-adaptive-strategy}.
	We can assume $\pi^+$ selects nodes in $\Gamma(\psi)$ at the beginning since they will eventually appear in the seed set regardless of the realization of the live-edge graph. 
	For $i \notin \Gamma(\psi)$, $\pi^+$ would select node $i$ as a seed with probability $x_i$,
	according to Definition~\ref{definition:optimal-adaptive-strategy}.
	Let $\Psi_i$ denote the partial realization when $\pi^+$ selects node $i$.
	Conditioned on $\Psi_i$, the selection of $i$ provides a marginal gain of $\Delta_{f}(i | \Psi_i)$
	for the influence spread.
	When we take its expectation over $\Psi_i$ and then multiply it with $x_i$, we obtain the
	overall marginal gain of selecting $i$ as a seed in policy $\pi^+$.
	When summing over all $i \notin \Gamma(\psi)$, together with the non-adaptive influence spread of
	seed nodes in $\Gamma(\psi)$, we thus obtain:
	\begin{align}
	\label{eq:in-arborescence-eq15}
	\fps\left(\bx \vee \bI_{\Gamma(\psi)}\right) = \sum_{i \notin \Gamma(\psi)}x_{i}\cdot \E_{\Psi_i}[\Delta_{f}(i | \Psi_i)] + \fa\left(\Gamma(\psi)\right).
	\end{align}
	Combining Eq.~\eqref{eq:in-arborescence-eq4}, \eqref{eq:in-arborescence-eq3}, \eqref{eq:in-arborescence-eq15}, it suffices to prove 
	\begin{align}
	\label{eq:in-arborescence-eq16}
	\Delta_{f}(i | \Psi_i) \leq \Delta_{f}(i | \psi)
	\end{align}
	for any $i \notin \Gamma(\psi)$ and any partial realization $\Psi_i$ such that $\psi \subseteq \Psi_i$. 
	This is exactly the adaptive submodularity of the influence reach function under the IC model
	with full-adoption feedback, which is given in Proposition~\ref{prop:infAdaptiveSubmodular}.
	Thus, the lemma holds.
\end{proof}

\lemInArborescenceBoundary*

\begin{proof}
	Consider any partial realization $\psi$ and any node $v \in \dom(\psi)$. 
	Take the unique directed path from node $v$ to the root $u$, let $\bar{v}$ denote the node on the path which is (i) contained in $\Gamma(\psi)$ and (ii) closest to the root $u$. Then we set $S = \{\bar{v}:v\in \dom(\psi)\}$. Clearly there is no directed edge from $\Gamma(\psi) \backslash S$ to $V\backslash \Gamma(\psi)$ and we have $|\partial(\psi)| \leq |S| \leq |\dom(\psi)|$.
\end{proof}

\lemInfluenceSpreadTwoHop*

\begin{proof}
	Fix any realization $\phi \sim \psi$, and then consider any node $v$ in $\Gamma(\Gamma(\psi), \phi) \backslash \Gamma(\partial(\psi), \phi)$.
	There must exist a directed path $P$ from $\Gamma(\psi) \backslash \partial(\psi)$ to $v$, and the path $P$ does not contain any nodes in $\partial(\psi)$. According to the definition of the boundary set $\partial(\psi)$, there is no directed path from $\Gamma(\psi) \backslash \partial(\psi)$ and to $V\backslash \Gamma(\psi)$ , unless it goes through a node in $\partial(\psi)$. Thus we conclude that $v \in \Gamma(\psi) \backslash \partial(\psi)$ and this gives proof for Eq.~\eqref{eq:in-arborescence-eq5}. 
	With Lemma \ref{lem:in-arborescence-boundary}, we have
	$\sigma(\partial(\psi)) \le \OPT_N(G, |\dom(\psi)|)$.
	Therefore, Inequality~\eqref{eq:in-arborescence-eq6} holds.
\end{proof}

\lemNearConcaveOpt*
\begin{proof}
	Let $\G_{N}(G, k)$ denote the non-adaptive greedy solution that select $k$ seed nodes.
	For $X \in \left\{0, 1, \cdots, n\right\}$, the greedy solution is $1 - 1/e$ approximate to the optimal solution, i.e.,
	\begin{align}
	\label{eq:in-arborescence-eq8}
	\OPT_{N}(G, X) \leq \frac{e}{e - 1}\G_{N}(G, X).
	\end{align} 
	We note that the greedy solution $\G_N(G, X)$ is concave in $X$, due to the submodularity of the influence spread function. Then taking expectation over both sides of Eq.~\eqref{eq:in-arborescence-eq8}, by Jensen's inequality, we have
	\begin{align}
	\E\left[ \OPT_N(G, X) \right] &\leq \frac{e}{e - 1} \E[\G_N(G, X)] \leq \frac{e}{e - 1} \G_N\left(G, \E[X]\right) \notag\\
	&\leq\frac{e}{e - 1}\OPT_{N}\left(G, \E[X]\right) = \frac{e}{e - 1}\OPT_N(G, k).	\label{eq:in-arborescence-eq9}
	\end{align}
	This concludes the proof.
\end{proof}

\section{Missing Proofs from Section~\ref{sec:adaptivity-gap-out-arborescence}}
\lemUpperAdaptiveOptimal*
\begin{proof}
	Let $\pi$ be any adaptive strategy satisfies $ \Pr_{\Phi\sim\cP}\left[i \in V(\pi, \Phi)\right] = x_i, i \in [n]$. Let $\Evt_{i}$ denote the event that node $1$ becomes active right after $\pi$ chooses node $i$. Furthermore, we use $\Evt_{i:}$ to denote the event that node $1$ become active right after $\pi$ chooses a node from $\{i, i + 1, \cdots, n\}$. We notice that events $\Evt_{1}, \cdots, \Evt_{n}$ are disjoint and we have 
	\begin{align}
	\label{eq:out-arborescence-eq2}
	\fps_{1}(x_1, \cdots, x_n) = \sum_{j = 1}^{n}\Pr\left[\Evt_j\right] = \sum_{j = 1}^{i}\Pr\left[\Evt_j\right] + \Pr\left[\Evt_{i+1:}\right] \, \forall i.
	\end{align}
	It is easy to see that
	\begin{align}
	\label{eq:out-arborescence-eq3}
	\Pr[\Evt_{i+1:}] \leq p_{i + 1},
	\end{align} 
	since the event $\Evt_{i+1:}$ can only happen when the node $i + 1$ can reach node $1$. Moreover, let $\mF_i$ denote the event that the policy $\pi$ selects the node $i$ before any nodes in $\{1, \cdots, i\}$ are active. Then we have for any $j \in [n]$,
	\begin{align}
	\label{eq:out-arborescence-eq4} 	
	\Pr\left[\Evt_{j}\right] 
	=	\Pr_{\Phi}\left[\Evt_{j} | \mF_j\right] \cdot \Pr_{\Phi}\left[ \mF_j\right] 
	\leq \Pr_{\Phi}\left[\Evt_{j} | \mF_j\right]\cdot \Pr_{\Phi}\left[j \in V(\pi, \Phi)\right]
	= p_j \cdot x_j.
	\end{align}
	The first equality holds since the event $\Evt_j$ can only happen when $\pi$ selects the node $j$ before any nodes $\{1, \cdots, j\}$ are active. Combining Eq.~\eqref{eq:out-arborescence-eq2}~\eqref{eq:out-arborescence-eq3}~\eqref{eq:out-arborescence-eq4}, we complete the proof.
\end{proof}

\lemMultilinearExtensionTelescope*
\begin{proof}
	Since the node $1$'s predecessors form a directed line, for any $i$ we have
	\begin{align*}
	&F_1(0, \cdots, 0, x_i, \cdots, x_n) -  F_1(0, \cdots, 0, x_{i +1}, \cdots, x_n) \notag\\
	=&p_i \cdot \left(F_i(0, \cdots, 0, x_i, \cdots, x_n) -  F_i(0, \cdots, 0, x_{i +1}, \cdots, x_n)\right) \notag \\
	=&p_i\cdot \left(1 - \left(1 - x_i\right)\left(1 - F_i(0, \cdots, 0, x_{i +1}, \cdots, x_n)\right) - F_i(0, \cdots, 0, x_{i +1}, \cdots, x_n) \right)\notag\\
	=&p_i x_i \left(1 - F_i(0, \cdots, 0, x_{i +1}, \cdots, x_n)\right).
	\end{align*}
	The first two equalities hold because the realization and the selection of nodes are independent.
\end{proof}

\section{Missing Proof from Section~\ref{sec:adaptivity-gap-bipartite}}

\thmAdaptivityGapUpperBipartite*

\begin{proof}
	For each node $u$, it suffices to prove that for any configuration $(x_1, \cdots, x_n)$, 
	\begin{align}
	\label{eq:adaptivity-bipartite-upper-eq1}
	F_{u}(x_1, \cdots, x_n) \geq \left(1 - \frac{1}{e}\right)\fps_{u}(x_1, \cdots, x_n),
	\end{align}
	where $F_u$ and $\fps_u$ are the same as defined in the proof of Theorem~\ref{thm:adaptivity-gap-out-arborescence}.
	We use $p_i$ to denote the probability that node $i$ can reach node $u$, then we have
	\begin{align}
	\label{eq:adaptivity-bipartite-upper-eq2}
	F_u(x_1, \cdots, x_n) = 1 - \prod_{i = 1}^{n}(1 - p_i x_i).
	\end{align}
	On the other side, let $\Evt_{i}$ denote the event that node $u$ become active right after the optimal policy $\pi^+$ chooses node $i$. We know that $\Pr[\Evt_i] \leq x_i \cdot p_i$ and thus we can conclude that 
	\begin{align}
	\label{eq:adaptivity-bipartite-upper-eq3}
	\fps_u(x_1, \cdots, x_n) = \sum_{i = 1}^{n}\Pr\left[\Evt_i\right] \leq \sum_{i = 1}^{n}x_i p_i.
	\end{align}
	Combining Eq.~\eqref{eq:adaptivity-bipartite-upper-eq2}~\eqref{eq:adaptivity-bipartite-upper-eq3} and the fact that 
	\begin{align}
	\label{eq:adaptivity-bipartite-upper-eq4}
	1 - \prod_{i = 1}^{n}(1 - y_i) \geq \left(1 - \frac{1}{e}\right)\min\{1, \sum_{i = 1}^{n}y_n\}
	\end{align}
	holds for all $y_{i}\in [0, 1]$, we can prove Eq.~\eqref{eq:adaptivity-bipartite-upper-eq1} and conclude the proof.
\end{proof}

\section{Missing Proofs and Further Discussions from Section~\ref{sec:adaptivity-gap-lower}}

\clmDirectLineAdaptive*

\begin{proof}
	Consider the following adaptive policy $\pi$: $\pi$ always selects the inactive node that is closest to the origin of the directed line, until it reaches the budget. Let $X_i$ ($i \in [k]$) denote the number of nodes that can be reached from the $i^{th}$ seed and let $X = X_1 + \cdots + X_k$. It is easy to see that $\E[\OPT_{A}(G, k)] \geq \fa(\pi) =\E[X]$. Let $Y_i \sim GE(1 - 1/t)$, i.e., $Y_i$ is a geometry random variable parametrized with $1 - 1/t$. $Y_1, \cdots, Y_k$ are independent and we know that $\E[Y_i]$ = t and $\Var[Y_i] = t^2 - t$. Our key observation is that $\E[X] = \E[\min\{Y_1 + \cdots + Y_k, kt\}]$. By Chebshev bounds, we have
	\begin{align}
	\label{eq:adaptivity-lower-eq1}
	\Pr\left[Y_1 + \cdots + Y_k < (1 - \epsilon/2)kt\right] \leq  \frac{4k(t^2 - t)}{\epsilon^2 k^2 t^2} \leq \frac{4}{\epsilon^2k} \leq \epsilon/2.
	\end{align}
	Thus we know that 
	\begin{align}
	\E\left[\min\{Y_1 + \cdots + Y_k, kt\}\right] &\geq \Pr\left[Y_1 + \cdots +  Y_k \geq (1- \epsilon/2)kt\right]\cdot \left(1 - \epsilon/2\right)kt\notag\\
	&\geq \left(1 - \epsilon/2\right)\cdot\left(1 - \epsilon/2\right)kt \geq \left(1 - \epsilon\right)kt\label{eq:adaptivity-lower-eq2}.
	\end{align}
	This concludes the proof.
\end{proof}

\clmDirectLineNonAdaptive*

\begin{proof}
	In the non-adaptive setting, for any node $u$ and seed set $S$, we define the distance between the node $u$ and the set $S$ as the distance between $u$ and the closest predecessor of $u$ in $S$. We know that the probability that the node $u$ is active only depends on the distance between $u$ and $S$. Let $N_i$ ($i \geq 0$) denote the set of nodes that has distance $i$ with $S$. Then we know that (i) nodes in $N_i$ are active with probability $(1 - 1/t)^{i}$, (ii) $N_0, N_1, \cdots N_{kt-1}$ are disjoint and $|N_i| \leq k$. 
	Now we have that $\sigma(S) = \sum_{i = 0}^{kt-1}(1 - 1/t)^{i}\cdot |N_i| \le \sum_{i = 0}^{t - 1}(1 - 1/t)^{i}\cdot k$.
	Thus, we can conclude that the optimal non-adaptive solution is to select $v_{11}, \cdots, v_{k1}$ as seeds and $\E[\OPT_{N}(G, k)] = \sum_{i = 0}^{t - 1}(1 - 1/t)^{i}\cdot k = (1 - (1 -1/t)^t)kt$.
\end{proof}

\noindent{\bf Discussion on Existing Approaches.\ \ }
In this paragraph, we give a hard instance showing that existing approaches cannot yield better-than-$2$ upper bounds on the adaptivity gap. The hard instance is exactly the directed line constructed in Theorem \ref{theorem:adaptivity-gap-lower-bound}, i.e., a directed line of length $kt$ and each edge is live with probability $1 - 1/t$. We use node $i$ to denote the $(i - 1)^{th}$ successor of the origin of the directed line, notice that the origin itself is denoted as node 1.

\noindent{\bf Multilinear Extension.\ \ } One common strategy is to use the multilinear extension as in \cite{seeman2013adaptive, asadpour2015maximizing}. In \cite{asadpour2015maximizing}, they consider the {\em stochastic submodular optimization} problem and prove that $\fps(x_1, \cdots, x_n) \leq \frac{e}{e - 1} F(x_1, \cdots, x_n)$ holds for any configuration $(x_1, \cdots, x_n)$.   We show that the ratio of $\fps(x_1, \cdots, x_n) / F(x_1, \cdots, x_n)$ can approach to 2 in our example. To be more specific, consider the configuration $(1, 1/t, \cdots, 1/t)$, we claim that $\fps(1, 1/t, \cdots, 1/t) = kt$. 
Consider the adaptive policy $\pi$ that always selects the inactive node that is closest to the origin of the directed line. The policy $\pi$ will select the first node with probability 1 and other nodes with probability $1/t$, since it will seed a node if and only if its incoming edge is blocked, this can happen with probability $1/t$. On the other side, we have $F(1, 1/t, \cdots, 1/t) \leq  F(1 - 1/t, 0, \cdots, 0) + F(1/t, \cdots, 1/t) \leq (1 - 1/t) t + F(1/t, \cdots, 1/t) \leq t + \frac{1}{2} kt + k$. The first inequality holds because of the DR-submodularity of the multilinear extension
and the third one holds because every node $u$ in the line is active with probability
\begin{align}
&\sum_{i = 0}^{\infty}\Pr\left[u \text{ is activated by its } i^{th} \text{ predecessor}\right]\notag\\ 
&\leq \sum_{i = 0}^{\infty} \frac{1}{t}\cdot\left(1 - \frac{1}{t}\right)^{i}\cdot\left(1 - \frac{1}{t}\right)^{i} = \frac{1}{t}\cdot \frac{1}{1 - (1 - 1/t)^2} = \frac{t}{2t - 1}.    \label{eq:adaptivity-lower-eq3}
\end{align}  
We conclude that when $t, k \rightarrow \infty$, $\fps(1, 1/t, \cdots, 1/t) / F(1, 1/t, \cdots, 1/t) \rightarrow 2$.

\noindent{\bf Random Walk Non-adaptive Strategy.\ \ } In \cite{gupta2016algorithms,gupta2017adaptivity,bradac2019near}, the authors consider the {\em adaptive stochastic probing} problem and they convert an adaptive policy to a non-adaptive policy by sampling a random leaf of the decision tree of the adaptive policy. 
Using our hard instance in the previous paragraph, we can show that this approach (i.e., random walk non-adaptive strategy) can give an upper bound of at most 2. 
To be more specific, we again consider the adaptive strategy $\pi$ and its corresponding non-adaptive strategy $\cW(\pi)$, where $\cW(\pi)$ picks a random leaf of the decision tree of the policy $\pi$. 
We are going to show that $\fps(1, 1/t, \cdots, 1/t)/\sigma(\cW(\pi))$ approaches to 2 asymptotically and it is sufficient to show that $\sigma(\cW(\pi)) \le t + k + \frac{1}{2}kt$. 
We imagine that node 1 appears in $\cW(\pi)$ with probability $1/t$ instead of 1, this is for ease of analysis and it will decrease the influence spread for at most $(1-1/t)\cdot t$ due to the submodularity of the influence spread function. 
For any node $u$, $u$ is activated by its $i^{th}$ predecessor (if it has one) when 
(i) the random seed set $\cW(\pi)$ does not contain nodes between $u$ and its $i^{th}$ predecessor (this happens with probability $\left(1 - \frac{1}{t}\right)^{i}$), 
(ii) its $i^{th}$ predecessor is included in the seed set (this happens with probability $\frac{1}{t}$) and 
(iii) node $u$ can be reached from its $i^{th}$ predecessor (this happens with probability $\left(1 - \frac{1}{t}\right)^{i}$). 
Moreover, we know that the above three events are independent in the non-adaptive setting, thus the probability that node $u$ is activated by the $i^{th}$ predecessor is $\frac{1}{t}\cdot\left(1 - \frac{1}{t}\right)^{i}\cdot\left(1 - \frac{1}{t}\right)^{i}$ and the probability that it is active is no more than $\frac{t}{2t - 1}$. This concludes our argument.

\end{document}